\documentclass[a4paper,UKenglish,cleveref, autoref, thm-restate]{lipics-v2021}

\usepackage{mathtools}
\usepackage{amsfonts}
\usepackage{amssymb}
\usepackage{fixltx2e}
\usepackage{algorithmicx}
\usepackage[noend]{algpseudocode}
\usepackage{url}
\usepackage[all,cmtip]{xy}
\usepackage{cleveref}
\usepackage[numbers]{natbib}

\usepackage{tikz}
\usepackage{tikz-cd}
\usetikzlibrary{automata, positioning}

\mathcode`<="4268
\mathcode`>="5269
\mathchardef\lt="213C
\mathchardef\gt="213E

\usepackage{wrapfig}
\usepackage{todonotes}

\newcommand{\eqdef}{\mathrel{\stackrel{\mathrm{def}}{=}}}

\newcommand{\coalgfg}{\mathbf{Coalg}_\mathsf{fg}}

\newcommand{\Set}{\mathbf{Set}}
\newcommand{\Conv}{\mathbf{Conv}}
\newcommand{\colim}{\varinjlim}
\newcommand{\coalg}{\mathbf{Coalg}}
\newcommand{\EM}{\mathbf{EM}}
\newcommand{\SMod}{\mathbf{SMod}}
\newcommand{\conv}{\textsf{conv}}
\newcommand{\aff}{\textsf{aff}}
\newcommand{\Se}{\mathcal{S}}
\newcommand{\M}{\mathcal{M}}

\newcommand{\s}{\mathbb{S}}

\newcommand{\N}{\mathbb{N}}

\newcommand{\R}{\mathbb{R}}

\newcommand{\T}{\mathcal{T}}
\newcommand{\two}{\mathbf{2}}

\newcommand{\D}{\mathcal{D}}

\newcommand{\eword}{\varepsilon}

\newcommand{\out}{\mathsf{out}}

\newcommand{\obs}{\mathsf{obs}}
\newcommand{\row}{\mathsf{row}}

\newcommand{\sem}[1]{[\![#1]\!]}

\title{A Coalgebraic Approach to Reducing Finitary Automata} 

\author{Keri D'Angelo}{Department of Computer Science, Cornell University, Ithaca, NY, USA}{}{}{}
\author{Alexandra Silva}{Department of Computer Science, Cornell University, Ithaca, NY, USA}{}{}{}%
\author{Gerco van Heerdt}{Droit Financial Technologies, London, UK}{}{}{}
\author{Leon Witzman}{Department of Computer Science, Cornell University, Ithaca, NY, USA}{}{}{}

\authorrunning{J. Open Access and J.\,R. Public} 

\Copyright{Jane Open Access and Joan R. Public} 



\category{} 

\relatedversion{} 

\EventEditors{John Q. Open and Joan R. Access}
\EventNoEds{2}
\EventLongTitle{42nd Conference on Very Important Topics (CVIT 2016)}
\EventShortTitle{}
\EventAcronym{CVIT}
\EventYear{2016}
\EventDate{December 24--27, 2016}
\EventLocation{Little Whinging, United Kingdom}
\EventLogo{}
\SeriesVolume{42}
\ArticleNo{23}

\begin{document}
\nolinenumbers
\maketitle

\begin{abstract}
Compact representations of automata are important for efficiency.  In this paper, we study methods to compute {\em reduced automata}, in which no two states accept the same language. We do this for   {\em finitary automata} (FA), an abstract definition that encompasses probabilistic and weighted automata. Our procedure makes use of Milius' locally finite fixpoint. We present a reduction algorithm that instantiates to probabilistic and $\s$-linear weighted automata (WA) for a large class of semirings. Moreover, we propose a potential connection between properness of a semiring and our provided reduction algorithm for WAs, paving the way for future work in connecting the reduction of automata to the properness of their associated coalgebras.

\end{abstract}

\section{Introduction}
\label{sec:intro}


Given a language $L$ it is natural to ask whether there exists a {\em small} automaton with a state that that accepts $L$. Whereas the notion of smallest is clear for deterministic automata, the same is not the case for other types of automata. A deterministic automaton is state-minimal iff all its states are reachable and no state is redundant; that is, no state accepts the same language as any other state. Hence, given a deterministic automaton we can minimize it by eliminating redundant states (and then take reachability from the initial state). Unfortunately, this property does not hold for more general automata: eliminating redundant states does not necessarily yield a minimal probabilistic or non-deterministic automaton, though it does yield a smaller automaton, which we call {\em reduced} and which is interesting to study as procedures to compute reduced automata are simpler than minimization procedures (see e.g. ~\cite{kiefer} for probabilistic automata). In particular, we will look at reduction procedures for {\em finitary automata}, an abstract definition that encompasses probabilistic automata (PA) and weighted automata (WA). 

Let $\T: \Set \to \Set$ be a finitary monad. Finitary automata are automata that have an output function $\out: S \to O$ and a transition function $\delta : S \to \T(S)^A$. If we assume that $O\cong \T(Y)$ for some finite set $Y$, we can lift the output and transition maps to operate on elements of $\T(S)$---$\out^\sharp : \T(S) \to O$ and $\delta^\sharp : \T(S) \to \T(S)^A$---and this is useful in defining the semantics of the automata, which we will detail later.  
Typically, automata also include an initial state. Here, we do not include the initial state in the definition and will study the semantics parametric on a state. Therefore, we simply write finitary automata as pairs $\mathcal A = (S,<\out,\delta>)$ (these are coalgebras for the functor $O\times \T(-)^A$). 

Finitary automata are a direct generalization of non-deterministic automata: the difference between them is the type of {\em next} state. In one case, it is a non-deterministic choice, whereas in the other it is a {\em combination}. By combination, we mean an element of $\T(X)$. These notions are instances of {\em side-effects} captured by a monad~\cite{moggi}. PAs and WAs are special cases by taking specific monads to capture these side-effects (see \cref{sec:prelim}). Besides PAs and WAs, some other types of automata that also have this abstract view include nominal and quantum automata.

In this paper, we focus on the notion of a {\em reduced automaton}. Given an automaton, we say it is {\em reduced} if no state is redundant (what we call reduced is called observable in \cite{airbib}). In the cases of PAs and WAs, the notion of a state being redundant is more subtle than for deterministic automata: a state is redundant if its language is a combination of the languages of other states; i.e. for PAs, a  {\em convex} combination, and for WAs, a {\em linear} combination. A deterministic automaton that is reduced will be minimal when eliminating unreachable states. This does not hold for finitary automata: it is possible that an automaton may be reduced, and yet once we eliminate unreachable states from a specific start state, we obtain an automaton that is not state minimal (see \cref{ex:redvsminimal}).


We study the problem of reducing finitary automata coalgebraically. As described earlier, finitary automata can be lifted from $\Set$ to $\EM(\T)$, where $\Set$ is the category of sets and set functions, and $\EM(\T)$ is the Eilenberg-Moore category over the monad $\T$. If there exists a notion of a {\em base} for all free and finitely generated elements of $\T(X)$ and there exists an algorithm to find such a minimal base, then a reduced automaton can always be calculated. We use the word base rather than basis since we do not assume the base is unique nor does every element need to be able to be written as a {\em unique} combination of the base. It becomes clear why this base condition is necessary in \cref{sec:wa}. 

This type of work has been done for deterministic automata which have state spaces in $\Set$.  To reduce deterministic automata, one can take the epi-mono factorization of the map from the state space $X$ into the carrier of the final coalgebra $\two^{A^*}$ in $\Set$, where $\two$ represents the two element set. This approach cannot be taken for $\EM(\T)$ categories in general, where quotients are not always well-understood. In this paper, we present a generalized algorithm for finitary automata in $\EM(\T)$, and explicitly describe the procedure for the special cases of probabilistic automata and $\s$-linear weighted automata. 

The coalgebraic perspective shows that reduced automata can be viewed as an image in a canonical object, $\vartheta F$, the locally finite fixpoint of a functor $F$ (\cref{def:lfp}): it can be constructed in an analogous fashion to the determinization of non-deterministic automata followed by a minimization procedure in a category with extra algebraic structure. The locally finite fixpoint is the colimit of the inclusion functor from the category of coalgebras with finitely generated carrier into the full coalgebra category, and under certain assumptions, the locally finite fixpoint is a subcoalgebra of the final coalgebra. We show finitary automata satisfy these assumptions. In the case of probabilistic automata, this extra structure gives rise to procedures in $\EM(\D$), the category of convex algebras and convex maps. In the case of $\s$-linear weighted automata, this extra structure gives rise to procedures in $\EM(\Se_\s$), the category of semimodules and semimodule homomorphisms (see \Cref{ex:monad} and \Cref{ex:algebras}).

While both reduced and state-minimal automata offer gains in space efficiency for representing languages, we argue in this paper that though state minimal automata are perhaps the most obvious small representation to study, there are advantages in considering reduced automata as compact acceptors of languages. In particular, it is known for probabilistic automata that minimization is NP-hard~\cite{kiefer} whereas the decision problem for reduced PA is in polynomial time (this is a consequence of \Cref{thm:pol}).

Finally, we begin forming a connection between properness and reducibility. Properness of semirings was first introduced in \cite{simequiv}, but has more recently been extended to describe functors \cite{fixpoint}. If a functor is proper, the rational fixpoint, which is similar to the locally finite fixpoint used in this paper, is a subcoalgebra of the final coalgebra. This allows for understanding the behavior of automata. 


In a nutshell, the main contributions (and structure) of the paper are as follows: 

\begin{enumerate}

\item In \Cref{sec:minimality}, we prove a concrete formula for $\vartheta F$ where $F$ is finitary over $\EM(\T)$ that enables its computation in $\Set$---this can be seen as one of the main advantages of the coalgebraic outlook on the problem offering guidance on how to generalize reduction to more types of automata;

\item A characterization of a reduced finitary automaton by viewing its image in the locally finite fixpoint $\vartheta F$ is shown in \Cref{sec:minimality};

\item Two special cases of this framework are PAs and $\s$-linear WAs (for a large class of semirings $\s$). In \Cref{sec:examples}, we give the complete construction of these procedures. For probabilistic automata, we provide a complete algorithm on how to explicitly compute the reduced automaton as well as a result (\Cref{thm:pol}) that implies the decision problem for reduced PA is in polynomial time.

\item We discuss directions for future work in \cref{sec:future}, including a detailed explanation on how our work connects to proper semirings/functors. 
\end{enumerate}

\section{Preliminaries}\label{sec:prelim}




We assume basic knowledge of category theory (functors and natural transformations). We first begin with some facts on monads and then coalgebras.

\begin{definition}[Monad]
A monad is a triple $(T, \eta, \mu)$ where $T \colon \mathbf {C} \to \mathbf {C}$ is a functor on a category $\mathbf {C}$, $\eta \colon X \to TX$ and $\mu \colon TTX \to TX$ are natural transformations (called the unit and multiplication, respectively) satisfying $\mu \circ \eta = id = \mu \circ T\eta$ and $\mu \circ \mu = \mu \circ T \mu$. 

\end{definition}

$(T, \eta, \mu)$ is a finitary monad if the functor $T$ commutes with filtered colimits.  

 \begin{example}[Monads]\label{ex:monad}
 We give examples of three monads on $\Set$. 
 \begin{enumerate}
 \item 
 The powerset monad is given by $\mathcal P(X) = \{ U \mid U \subseteq X\}$, $\eta(x) =\{x\}$, and $\mu (\Phi) = \bigcup_{S\in \Phi} S$. 
 \item 
 The (finitely supported) distribution monad is given by $\D(X) = \left\{ \phi \mid \sum \phi(x)=1 \right\}$  (from \cref{dist}), with natural transformations
 $$
 \eta(x) =\lambda y. \begin{cases} 1 & y=x\\ 0 & \text{otherwise}\end{cases}\qquad\qquad
 \mu (\Phi)(x) = \sum\limits_{\varphi\in \mathbf{supp}(\Phi)} \varphi(x) \times \Phi(\varphi).
 $$
 \item The (finitely supported) free semimodule monad for a semiring $\s$ (with unity) is given by $\Se_\s(X) = \{ \varphi : X \to \s \mid \sum_{i=1}^n \varphi(x_i) \cdot x_i \}$ (from \cref{semi}), with natural transformations
 $$ \eta(x)  =  \lambda y. \begin{cases} 1 & y=x\\ 0 & \text{otherwise}\end{cases}\qquad\qquad
 \mu (\Phi)(x) = \sum\limits_{\varphi\in \mathbf{supp}(\Phi)} \varphi(x) \times \Phi(\varphi).$$
 \end{enumerate}
 \end{example}

$\mathcal P$, $\D$, and $\Se_\s$ are all examples of finitary monads. 

 A {\em coalgebra} is a pair $(X, t\colon X \to \mathcal F X)$, where $X$ is the {\em state space} and $t$ is the transition dynamics which is parametric on a functor $\mathcal F \colon \mathbf{C} \to \mathbf{C}$. Coalgebras cover a range of automata types: deterministic automata are coalgebras for the functor $D(X) = 2\times X^A$ on $\Set$ (the category of sets and functions), weighted automata (over a field $\mathbb{F})$ are coalgebras for the functor $W(V) = \mathbb F \times V^A$ on $\mathbf{Vect}$ (the category of vector spaces and linear functions), and as we will see probabilistic automata can be seen as coalgebras in $\Conv$ (the category of convex sets and convex maps). The advantage of studying different automata as coalgebras is that many important notions are fully determined by $
 \mathcal F$, e.g homomorphisms, behavioral equivalence. More recently, there's been a research effort in showing that $\mathcal F$ is also rich enough to design abstract algorithms, including minimization algorithms~\cite{DBLP:conf/ifipTCS/KonigK14,milius-fm,DBLP:journals/lmcs/WissmannDMS19,DBLP:conf/tacas/BirkmannDM22}. An important concept in coalgebra is {\em finality}:  a coalgebra $(Z,z \colon Z\to \mathcal F Z)$ is {\em final} if for all 
 
  \begin{wrapfigure}{r}{.22\textwidth}\vspace{-30pt}
 \[
\hspace*{-.4cm} \xymatrix@R=.45cm@C=1.2cm{
X\ar[d]_f \ar@{-->}[r]|{\ \ \sem -\ \ }& Z\ar[d]^z\\
 \mathcal F X\ar@{-->}[r]&\mathcal F Z
 }
 \]\vspace{-25pt}
 \end{wrapfigure}
\noindent other coalgebras $(X,f \colon X\to \mathcal F X)$ there exists a unique structure-preserving homomorphism $\sem - \colon X \to Z$, as depicted on the right.

 Final coalgebras are an abstract way of capturing behavior: any state on another coalgebra can be mapped uniquely into it. Hence, $\sem x$ is a canonical representative of the behavior denoted by $x$. For concrete functors, this instantiates to familiar things: 
 \begin{example}[Final coalgebras]\label{ex:final}
 We give examples of final coalgebras in $\Set$ and $\mathbf{Vect}$. 
 \begin{enumerate}
 \item The final coalgebra of $D(X) = 2\times X^A$ on $\Set$ is the pair $(2^{A^*},  <\varepsilon?,\partial>)$, with the set of languages over $A$ as the carrier and the transition structure given by language derivatives: 
 \[
 \varepsilon?(L) = \begin{cases} 1 & \varepsilon\in L\\ 0 & \text{otherwise}\end{cases}
 \qquad \qquad \qquad
 \partial(L)(a)= \{w \mid aw \in L \}
 \]
 $\sem -$ assigns to a state $x$ of a deterministic automaton the regular language accepted by $x$. 
 \item The final coalgebra of $W(V) = \mathbb R \times V^A$ on $\mathbf{Vect}$ is the pair $({\mathbb R}<<A^*>>,  <o,t>)$, with the set of weighted languages (formal power series) over $A$ as the carrier and the transition structure given by the linear maps: $o(\sigma) = \sigma(\varepsilon)$ and $t(\sigma)(a)(w)= \sigma(aw)$. 
 $\sem -$ assigns to a state $x$ the rational power series denoting the weighted language accepted by $x$. 
 \end{enumerate}
 \end{example}
 The two examples of final coalgebras above are actually of similar nature (hence the closeness in their definitions). The second example is in the category of vector spaces and linear maps, which is an instance of a category of algebras for a monad, a concept we recall next (the first example is also an instance albeit for the simple identity monad!).

   \medskip
 
 Given a monad $T$, an algebra for $T$ (or $T$-algebra) is a pair $(X, h\colon TX \to X)$ where $h$, the algebra map, satisfies $h\circ \eta = id$ and $h \circ T h = h \circ \mu$. The category of algebras for a monad $T$, also called the category of {\em Eilenberg-Moore algebras for $T$}, denoted \textbf{EM}($T$), has $T$-algebras as objects and structure preserving maps morphisms. 

\begin{example}[Algebras for a Monad]\label{ex:algebras}
 We instantiate $T$-algebras for the monads of \cref{ex:monad}. 
 \begin{enumerate}
 \item $\EM(\mathcal P$) is the category of join-semilattices and join-preserving maps. 
  \item  $\EM(\mathcal D$) is the category of convex sets and convex maps.   
  \item $\EM(\Se_\s)$ is the category of semimodules over the semiring $\s$ and $\s$-linear maps.
   \end{enumerate}
 \end{example}
 
The second example above could also be presented with affine maps, as we have the following result from \cite{poly}: 
For a convex set $P$ and any $Q \subseteq E$ for $E$ a real coordinate space, $f$ preserves convex combinations iff $f$ preserves affine combinations.

 We can now show how the two examples in \cref{ex:final} are in fact an instance of the same result~\cite{DBLP:journals/jcss/Jacobs0S15}:
 
 \begin{theorem}\label{thm:final}
Let $T$ be a monad (in $\Set$) and $O$ be an algebra for $T$. The final coalgebra of the functor $M(X) = O \times X^A$, where $M\colon \mathbf{EM}(T) \to \mathbf{EM}(T)$, is the set of $O$-weighted languages $O^{A^*}$, is as follows:
\[
 \xymatrix@R=.45cm{
 X\ar[d]_f \ar@{-->}[r]^{\sem -}\ar@{}[dr]|*+<7pt>[o][F]{\star}& O^{A^*}\ar[d]^z_\cong\\
 O\times  X^A \ar@{-->}[r]&O\times (O^{A^*})^A
 } \qquad z(\varphi) = <\varphi(\varepsilon), \lambda a.\lambda w. \varphi(aw)>
\]

 \end{theorem}

\section{Characterizing Reduced Finitary Automata}\label{sec:minimality}

The type of automata we study in this paper are what we call finitary automata. Finitary automata use a finitary monad in their transition structure. For the remainder of this paper, we assume the monad $\T$ to denote a finitary monad on $\Set$.
        
\begin{definition}[Finitary Automata (FA)]
A finitary automaton is a pair $(Q, <\delta, \out>)$, consisting of a set of states $Q$, a transition function $\delta \colon Q \to \T(Q)^A$, and an output function $\out \colon Q \to O$, where $\T : \Set \to \Set$ is a finitary monad, and $O\cong\T(L)$ for some set $L$.
\end{definition}

\begin{remark}\label{remark:base}
Later, we will further assume the monad $\T$ satisfies a {\em base condition}, which we define to mean that for all free and finitely generated carriers in $\EM(\T)$, that is, the carrier is of the form $\T(X)$ for some finite set $X$, there exists an algorithm to find a {\em base} (a minimal generating set, not necessarily unique) $\overline X \subseteq X$ such that $\T(\overline X) \cong \T(X)$. We will see an example in \cref{sec:wa} why this base condition is necessary.
\end{remark}

Finitary automata are coalgebras for the $\Set$ functor $O\times\T(-)^A$. In order to understand how finitary automata work as acceptors of languages we use a coalgebraic generalization of the usual subset construction for non-deterministic automata~\cite{DBLP:journals/jcss/Jacobs0S15} in the following way:
 \[
 \xymatrix@R=.5cm{
 X \ar[d]_{<\out,\delta>}\ar[r]^{\eta_X}& \T (X)\ar[dl]^-{<\out^\sharp,\delta^\sharp>}\ar@{}[dr]|*+<7pt>[o][F]{\star} \ar@{-->}[r]^{\sem -}& O^{A^*}\ar[d]_\cong\\
 O \times \T(X)^A \ar@{-->}[rr] && O\times  {(O^{A^*})}^A \
 }
\]
Here, $\eta$ is the monad unit and, because $O\cong T(L)$, the coalgebra $<\out, \delta> \colon X \to O \times \T(X)^A$ in $\mathbf{Set}$ can be lifted to a coalgebra $<\out^\sharp, \delta^\sharp> \colon \T(X) \to O \times \T(X)^A$ in $\mathbf{EM}(\T)$ (see e.g. \cite{DBLP:journals/jcss/Jacobs0S15}). 

The commutativity of $\xymatrix{*+<7pt>[o][F]{\star}}$ comes from \Cref{thm:final} and therefore $\sem -$ is actually a homorphism in $\mathbf{EM}(\T)$. For intuition, if we replace $\T$ by the powerset monad in the above diagram (with $\eta(x)=\{x\}$) and $O$ by $\mathbf{2}$, the two-element set, the above automata would coincide with non-deterministic automata (NFA) and both the final state and transition functions would correspond to the well-known subset construction: e.g. $\out^\sharp(U) = 1 \iff \exists u\in U.\ \out (u) =1$. The map $\sem{-}$ would assign to state $\{x\}$ the language it accepts (as defined for an NFA).

We will often use $L_x$ to denote the language accepted by a state, which is $\sem{\eta(x)}$.

Minimal deterministic automata have a universal property: they are unique (up-to isomorphism) for a given language. In the cases of non-deterministic and probabilistic automata, there are non-isomorphic state minimal automata accepting the same (probabilistic) language. For this reason, algorithms to find minimal probabilistic automata are hard(er) to design. Our idea for developing an approach to reduction of automata stemmed from wanting a more efficient way to reduce the size of the state space. We will take a coalgebraic perspective which allows us to generalize reduction beyond just these two types of automata, namely, to finitary automata. The goal of this section is to show an alternative to minimal automata, which we will call {\em reduced automata}, and show reduced automata have a universal property similar to minimal deterministic automata.

\begin{definition}[Redundant State]
Let $X$ be the set of states. A state $x \in X$ is considered redundant if there exists some $y\in \T(X \backslash \{x \})$ such that $x$ and $y$ accept the same language; that is, $\sem{\eta(x) }= \sem y$.
\end{definition}

\begin{definition}[Reduced automata]
 An \emph{automaton} $\mathcal A = (Q,<\out,\delta>)$ is {\em reduced} if there are no redundant states.
 \end{definition}
 
  For example, a \emph{probabilistic automaton} $\mathcal A = (Q, <\delta, \out>)$ is {\em reduced} if there is no $q\in Q$ such that, for all $w\in A^*$
  \[
  L_q (w) = \sum_{s\in Q\setminus \{q\}} r_s \times L_s(w) \qquad\qquad \text{ and } \sum_{s\in Q\setminus \{q\}} r_s = 1 , \quad r_s \in [0,1]
  \]
That is, $L_q$ is not a convex combination of languages of other states in $Q$.


\begin{example}[Minimal vs reduced]\label{ex:redvsminimal}
Clearly, minimal automata are reduced but not the other way around. We give an example of a PA over the alphabet $A=\{a\}$. We depict the transitions and output functions as a state diagram---missing outputs are $0$. 

    
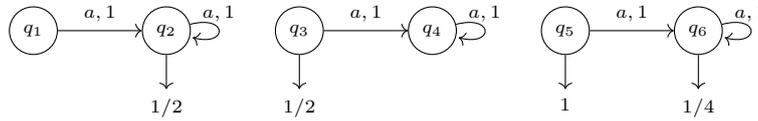
\begin{tikzpicture}[node distance = 1.75cm, arrows=->,font=\scriptsize]
    \tikzset{every state/.style={minimum size=0pt}}
        \node[state] (q1) {$q_1$};
        \node[state, right of=q1, accepting below, accepting text = {$1/2$}] (q2) {$q_2$};
        \node[state, right of=q2, accepting below, accepting text = {$1/2$}] (q3) {$q_3$};
        \node[state, right of=q3] (q4) {$q_4$};
        \node[state, right of=q4, accepting below, accepting text = {$1$}] (q5) {$q_5$};
        \node[state, right of=q5, accepting below, accepting text = {$1/4$}] (q6) {$q_6$};
        \draw (q1) edge[above] node[above]{$a,1$} (q2)
              (q2) edge[loop right] node[above]{$a,1$} (q2)
              (q3) edge[below] node[above]{$a,1$} (q4)
              (q4) edge[loop right] node[above]{$a,1$} (q4)
              (q5) edge[below] node[above]{$a,1$}(q6)
              (q6) edge[loop right] node[above]{$a,1$} (q6);
    \end{tikzpicture}

This automaton is not reduced. The language of state $q_6$ can be written using the languages of $q_2$ and $q_4$ (or $q_1$ and $q_3$): $L_{q_6} = 0.5 L_{q_2} + 0.5 L_{q_4}$. Eliminating state $q_6$ would result in a 5 state automaton that is reduced---note how the choice of writing $L_{q_6}$ as a convex combination results in two different transition structures; this illustrates that reduced automata are not unique. The reduced automaton is not minimal: for example if we fix $q_5$ as the initial state, we can see that there is an automaton with $2$ states accepting $L_{q_5}$. 
\end{example}

In the next sub-section we will provide a characterization of the reduced automaton as the image of a unique map in an Eilenberg-Moore category, and this will show a way to develop algorithms to compute this reduced automaton. 

In \cref{sec:lfp} and \cref{sec:fp_red} we show the following: first, we compute the image of a finitary coalgebra $(\T(X), \beta)$ in the final coalgebra. This utilizes a certain colimit, called the locally finite fixpoint (\cref{def:lfp}). We then show if there exists an algorithm to compute a minimal base for $\T(X)$, we can reduce the automata. This algorithm will come into play in the choice of the maps in the colimit we compute (\Cref{eq:colim}). This will relate the state space of the original automaton $\T(X)$ and the reduced automaton $\T(\overline X)$ in the final coalgebra. Due to some of the nice properties we have that are discussed below, we are able to compute the colimit in $\Set$ and then revert back to $\EM(\T)$.

  \subsection{Reduced automata as subcoalgebras of the locally finite fixpoint}\label{sec:reducedlfp}

Our goal here is to ``factor'' the morphism 
\begin{equation}\label{eq:factor}
\sem - = \xymatrix{ \T(X) \ar@{->>}[r]^-e& I \ar@{>->}[r]^-i& O^{A^*}} 
\end{equation}
where this factorization extends to the coalgebra structure so $I$ can be given an automaton structure. Because $i$ is a mono, we can see that $I$ does not contain two states that can be mapped to the same language. 

We first solve the following questions: When and how can we effectively compute $I$ and its transition structure? Can such an $I$ be represented as an automaton in $\Set$ rather than in $\EM(\T)$: in other words, when does there exist a set $Y$ such that $I\cong \T(Y)$? This representation in $\Set$ will be exactly the reduced automaton as the language of a state in $Y$ will not be a combination of other states, since $i$ is monic.

\paragraph*{Computing $I$: Epi-mono factorization in $\EM(\T)$.}

Deterministic automata have state spaces in $\Set$. Thus, when looking at language equivalence of two states, one can easily compute the epi-mono factorization of the final map $\sem - \colon Q \to \two^{A^*}$. This is because quotients are well-understood in $\Set$. When viewing $\EM(\T)$ in general, this same notion of quotient does not translate since $\EM(\T)$ may be neither abelian nor additive. For example, for $\D$ the finitely supported distribution monad (\cref{dist}), $\EM(\D)$ is neither abelian nor additive. Thus actually computing the epi-mono factorization in \cref{eq:factor} is no longer as straightforward as taking a quotient, since the notion of a quotient is not yet fully understood in $\EM(\T)$ (see e.g \cite{Gubeladze2019AffinecompactF}). This lack of clarity on how to obtain $I$ from computing $e$ in \cref{eq:factor} led us to take a different perspective, made possible by the coalgebraic outlook. We will compute $I$ using the fact that $I$ must be a subcoalgebra of the final coalgebra. Because all $\T(X)$ are finitely generated state spaces, we can look at its image in the {\em locally-finite fixpoint} (\cref{def:lfp}) of a finitary functor $F$ over $\EM(\T)$.  This gets us a kind of epi-mono factorization into the final coalgebra. We will study this fixpoint below and we will show how the explicit computation of the fixpoint can be used to eliminate the redundant states.

\subsection{Locally finite fixpoint}\label{sec:lfp}

We briefly recall some facts on the the locally finite fixpoint of a functor.  Below, it is assumed $\mathcal A$ is a locally finitely presentable (lfp) category. The goal is to calculate and use this fixpoint in determining reduction of finitary automata.

\begin{definition}[Locally finite FP \cite{fixpoint}]\label{def:lfp}

The {\em locally finite fixpoint} of a functor $F\colon \mathcal A \to \mathcal A$ is given as the colimit $(\vartheta F, \ell)=\colim(\coalgfg (F) \xhookrightarrow{i} \mathbf{Coalg} (F))$, where $i$ is the inclusion functor and $\coalgfg(F)$ consists of all $F$-coalgebras with finitely generated carriers.

\end{definition}


\begin{theorem}[\cite{fixpoint}]\label{thm:fixpoint}

Suppose that the finitary functor $F : \cal{A} \to \cal{A}$ preserves monos. Then $(\vartheta F, \ell)$ is a fixpoint for $F$, and it is a subcoalgebra of $v F$ ($v F$ denotes the final $F$-coalgebra).

\end{theorem}

To use the preceding theorem, we first show the coalgebra we are working with is a finitary functor.

\begin{theorem}[Lifting to $\EM(\T)$]\label{thm:fin}

$F: \EM(\T) \to \EM(\T)$ defined by $F(Q) = O \times Q^A$ is a lifting of a functor $F_0 : \Set \to \Set$ defined by $F_0(X)=O \times X^A$.
 
 \end{theorem}

 \begin{proof}
 By lifting of a functor, it is meant that $F_0 \circ U = U \circ F$, where $U: \EM(\T) \to \Set$ is the forgetful functor. $F_0 \circ U (Q, \alpha) = F_0 (Q) = O \times Q^A$. In the other direction, $U \circ F(Q, \alpha) = O \times Q^A$ in $\Set$ since we apply $F$ to the algebra $(Q, \alpha)$ and then the forgetful functor removes the algebra structure.
 
 \end{proof}
 
 \begin{theorem}\label{thm:finitary}
 
 $F_0$ is a finitary functor.
 
 \end{theorem}
 
\begin{proof}
The functor $F_0$ can be rewritten as $O \times (A,-)$. We want to show it preserves filtered colimits. $A$ is a finite alphabet and therefore finitely presented in $\Set$ so $(A,-)$ commutes with filtered colimits. Moreover, filtered colimits commute with limits in $\Set$, so we have $F_0$ preserves filtered colimits.
\end{proof}

Therefore, $F$ is a finitary functor \cite[Remark 3.2]{fixpoint}.

Now it can be seen that the category of $F$-coalgebras, $\coalg(F)$, we are considering in this paper satisfy all necessary assumptions. All categories considered here are of the form $\EM(\T)$ for a finitary monad $\T: \Set \to \Set$. Eilenberg-Moore categories over a finitary monad on $\Set$ are algebraic categories, and thus lfp. The assumption of $F$ preserving monos is not necessary since we showed in \cref{thm:fin} that $F$ is a lifting of a set functor \cite{fixpoint}. Thus, $\vartheta F$ for the finitary functor $F(Q)=O \times Q^A$ over $\EM(\T)$ is a subcoalgebra of the final coalgebra. 
To use this, we first calculate the colimit from \cref{def:lfp}: $(\vartheta F, \ell) = \colim(\coalgfg(F) \xhookrightarrow{i} \coalg(F))$. 
We use the fact that the forgetful functor $U\colon \coalg(F) \to \EM(\T)$ creates colimits (this holds for any forgetful functor $U\colon \coalg(F) \to \mathcal A$ \cite{milius}). This gives the following composition of morphisms
$$\coalgfg(F) \xhookrightarrow{i} \coalg(F) \xrightarrow{U} \EM(\T)$$
allowing the colimit to be computed in $\EM(\T)$. We want to apply a similar result for the forgetful functor $\bar U : \EM(\T)\to \Set$ to be able take the colimit in $\Set$, like as done for DFAs, which would lead to the following chain of morphisms:

$$\coalgfg(F) \xhookrightarrow{i} \coalg(F) \xrightarrow{U} \EM({\T}) \xrightarrow{\overline{U}} \Set$$

The forgetful functor $\bar U : \EM(\T)\to \Set$ creates colimits that are preserved by $\T$, and $\T$ preserves filtered colimits since it is finitary. Using the following result from \cite{effect} allows us to compute $\vartheta F$ in $\Set$.

\begin{theorem}[\cite{effect}]\label{thm:filtered}

$\vartheta F$ is a filtered colimit.

\end{theorem}

In this paper, we offer an explicit formula for $\vartheta F$ for finitary coalgebras. 

\begin{equation}\label{eq:colimit}
\vartheta F = \colim(\overline{U} \circ U \circ i) = \left(\bigsqcup_{((C,\alpha), \beta)\in \coalgfg(F)} \overline{U} \circ U \circ i ((C, \alpha),\beta) \right)\Big/{\sim}
\end{equation}

and hence

\begin{equation}\label{eq:colim}
\vartheta F = \left(\bigsqcup_{(C,\alpha) \in \EM(\T)_{\mathsf{fg}}} C \right)\Big/{\sim}
\end{equation}
where $\EM(\T)_{\mathsf{fg}}$ denotes the subcategory of finitely generated elements in $\EM(\T)$ and $\sim$ is the smallest equivalence relation on the coproduct in $\Set$; i.e. the disjoint union, such that $x \sim \overline{U} \circ U \circ i (f (x))$ for all $f : (C, \alpha) \to (D, \beta)$ in $\coalgfg(F)$ and all $x\in C$. Moreover, since $\coalgfg(F)$ is filtered, we have that the equivalence relation is precisely
\begin{equation}\label{eq:sim}
$$(x \in C) \sim (y \in D)\text{ $\iff$  }\exists E, (f : C \to E), (g : D \to E)\text{ such that }f(x) = g(y).$$
\end{equation}

\subsection{Fixpoints for Reduction}\label{sec:fp_red}
We now use the explicit description of $\vartheta F$ (\Cref{eq:colim}) and focus on the image of $\T(X)$ inside $\vartheta F$: 
 \begin{equation}\label{eq:general}
 \xymatrix@C=.75cm@R=.25cm{
 X \ar[dd]_{<\out,\delta>}\ar[r]|{\ \eta_X\ }& \T (X)\ar[dr]_{q\circ i_{\T(X)}}\ar[ddl]^-{<\out^\sharp,\delta^\sharp>}  \ar@{-->}[rr]|{\ \ \sem -\ \ }&&O^{A^*}\ar[dd]_\cong\\
 && \left({\bigsqcup C}\right)\Big/{\sim}\ar@{>->}[ur]\ar[d]\\
 O\times \T(X)^A \ar@{->}[rr] &&O\times ( \left({\bigsqcup C }\right)\Big/{\sim})^A\ar@{>->}[r]& O\times  {(O^{A^*})}^A \
 }
\end{equation}
Here, $i_{\T(X)}$ represents the inclusion of $\T(X)$ into the coproduct and $q$ is the projection onto the quotient. As $\T(X)$ generates the state space of an automaton with set of states $X$, we need a notion of quotient to be able to identify redundant states, which is where we make use of the locally finite fixpoint.

We first take the image of $\T(X)$ into the coproduct $\bigsqcup C$ (in $\Set$). $C$ represents the underlying set of a finitely generated element in $\EM(\T)$. We now use \Cref{eq:colim} to see that the image of $(\T(X), <\out^\sharp,\delta^\sharp>)$ under $\sem -$, which we denote by $\textsf{im}(\T(X))$, lies in $\bigsqcup C / \sim$. The goal now is to reduce the set of states $X$ by using the definition of $\sim$.

First note that any morphism in $\EM(\T)$ with domain of the form $\T(X)$ is completely determined by where it sends $X$. We use the definition of $\sim$ to be consistent with that of filtered colimits since $\vartheta F$ is filtered. Since we are viewing $\textsf{im}(\T(X))$ inside $\vartheta (F)$, $f$ and $g$ in the definition of the colimit (\Cref{eq:sim}) are restricted to having (co)domain $\T(X)$. 

It is at this point where our reason for needing a base becomes necessary. In order to proceed with determining what $f$ and $g$ will be in \cref{eq:sim}, we must calculate a base of $\T(X)$ in $\EM(\T)$. We assume an existing algorithm for calculating $\overline X$, the {\em base} of $\T(X)$ (\Cref{remark:base}). In other words, $\T(X)$ generates the same language as $\T(\overline X)$. Choose the map $f : \T(X) \to \T(X)$ to be defined as $f(x)=x$ for all $x\in \overline X$. For $x\in X\backslash \overline{X}$, choose $y\in \T(X\backslash \{x \})$ such that $y$ accepts the same language as $x$ and define $f(x) = y$. Note that the definition of base assumes a choice may need to be made here, that is, there may be more than one $y\in \T(X\backslash \{x \})$ that accepts the same language as $x$, resulting in a reduced automaton that may not be unique. 

We note that the image of $f$ is contained within $\T(\overline X)$. For the function $g$, we choose $g$ to be the identity map on $\T(X)$. This will reduce the set $X$ by keeping all states $x\in \overline{X}$ as $f$ relates them to themselves, and relating redundant states in $X$ to a combination of states in $\overline{X}$. Both $f$ and $g$ were designed so they can be lifted to the category $\EM(\T)$. 

\paragraph*{Correctness} 
We show two things here. The first is that $\T(X) \sim \T(\overline X)$ for a base $\overline X$, and the second is no smaller $Y \subsetneq \overline{X}$ satisfies $\T(X) \sim \T(Y)$; that is, $\overline X$ is a smallest such set that generates the same language.

Consider the following diagram.  
\begin{equation}\label{correctness}
\begin{tikzcd}
\T(X) \arrow[r, "f"] \arrow[d, "c"'] & \T(\overline{X}) \arrow[d, "\overline{c}"'] & \T(\overline{X})  \arrow[l, "\mathbf{1}"] \arrow[d, "\overline{c}"] \\
F(\T(X)) \arrow[r, "F(f)"'] & F(\T(\overline{X})) & F(\T(\overline{X}))  \arrow[l, "F(\textbf{1})"]
\end{tikzcd}
\end{equation}
where $\mathbf{1}$ represents the identity morphism and $f$ is the same $f$ as defined at the end of the previous subsection. Since $\T(\overline{X}) \subseteq \T(X)$, $\overline{c}$ is simply a restriction of $c$. Using that $g$ here is the identity map, we now show the diagram is commutative.

\begin{theorem}
The diagram in \Cref{correctness} is commutative.
\end{theorem}
\begin{proof}
For the left square,
\begin{equation} F(f) \circ c(x) = F(f) < \out(x) , \delta_x > = <\out(f(x)), \delta_{f(x)} > = \overline{c}(f(x))\end{equation}
For the right square, \begin{equation} F(\mathbf{1}) \circ \overline{c}(\overline x) = F(\mathbf{1}) <\out(\overline x), \delta_{\overline x}> = <\out(\overline x), \delta_{\overline x}> = \overline c \circ \mathbf{1} (\overline x)\end{equation} \end{proof}

We next show $\T(\overline X)$ corresponds to the reduced state space; that is, there are no smaller state spaces that will generate the same language.

\begin{theorem}
If $Y \subsetneq X$, then $\T(Y) \not\sim \T(X)$. 
\end{theorem}

\begin{proof}
To show this is the {\em most} reduced state space, assume we have some $Y \subsetneq \overline{X}$. We will show $\T(Y)  \not\sim \T(X)$. 

Assume that there does exist some $Y \subsetneq \overline{X}$ such that $\T(Y) \sim \T(X)$. This would mean we have the following commutative diagram.  

\begin{center}
\begin{tikzcd}
\T(X) \arrow[r, "f"] \arrow[d, "c"'] & C \arrow[d, "e"'] & \T(Y)  \arrow[l, "g"] \arrow[d, "d"] \\
F(\T(X)) \arrow[r, "F(f)"'] & F(C) & F(\T(\overline{X}))  \arrow[l, "F(g)"]
\end{tikzcd}
\end{center}
where $C$ is a finitely generated object in $\EM(\T)$ and $f$ and $g$ are maps in $\EM(\T)$. Since $Y \subsetneq \overline X$ and $\overline X$ is a base for the state space, this would imply there exists some $x\in \T(X)$ and $y\in \T(Y)$ such that $f(x)=g(y)$ but $x$ and $y$ do not generate the same language. Moreover, since \Cref{thm:fixpoint} says $\vartheta F$ is a subcoalgebra of the final coalgebra, $f(x)=g(y)$ implies $x\sim y$ and thus $\sem{x}= \sem y$, a contradiction. Therefore, there exists no $Y \subsetneq \overline{X}$ such that $\T(X) \sim \T(Y)$.
\end{proof}

\section{Examples}\label{sec:examples}

In this section, we give two examples of finitary automata. We give explicit constructions of the locally finite fixpoint for each example and utilize algorithms to define the map $f$ in each example's respective colimit.

\subsection{Probabilistic Automata}

Our first example of finitary automata is probabilistic automata. 

\begin{remark}
In \cref{sec:future}, we will make a connection with the work done in \cite{proper}. We note that in this section, although we work over the category $\EM(\D)$, the category of convex sets, everything in this section can easily be translated to $\EM(\D_s)$, the category of positive convex algebras ($\mathbf{PCA})$. ($\D_s$ denotes the finitely supported sub-distribution monad). 
\end{remark}

\begin{definition}[Distributions]\label{dist}

Let $\D : \Set \to \Set$ denote finitely supported distributions on a set $X$, defined as 
$\D(X) = \left\{ \phi: X \to [0,1] \mid \sum \phi(x)=1 \right\}$ where the support of each $\phi$ is finite. The support of a distribution is the set $\textsf{supp}(\phi)= \{ x \mid \phi(x)\neq 0\}$.

\end{definition}

\begin{definition}[Probabilistic automaton (PA)]
    A \emph{probabilistic automaton} is a tuple $(Q, <\delta, \out>)$, consisting of a set of states $Q$, a transition function $\delta \colon Q \to \D(Q)^A$, and an output function $\out \colon Q \to [0, 1]$.
\end{definition}
The output function can be seen as a row vector (labelled by $Q$) with entries in $[0,1]$. $\delta$ can be equivalently represented as an $A$-indexed family of functions $\delta_a \colon Q \to \D(Q)$, each of which can be seen as square matrices (labelled by $Q$). Note that $\D(\two)=[0,1]$. Thus, a probabilistic automaton can be viewed as the coalgebra $F(X) = [0,1] \times X ^A$ on $\EM(\D)$, where $\D$ is clearly a finitary monad. Since $\EM(\D) \cong \Conv$, we use the two interchangeably for the remainder of this paper.

\paragraph*{Reverse determinization: from $\Conv$ to $\Set$.} 
If we begin with a convex set $C$, we want to determine when it can be written in the form $\D(Y)$ where $Y$ is a finite set. We characterize this in the following theorem; the proof is a consequence of \cite[Theorem 1]{basis}, stating that every finitely generated convex set has a unique base, and that base is precisely the extreme points of the convex set. 

\begin{theorem}

A convex set $C$ is isomorphic to $\D(Y)$ as convex sets for a finite set $Y$ if and only if $C$ is a simplex. 

\end{theorem}

\begin{proof}

Consider a simplex $C$ with the set of vertices $Y$. By definition of simplex, the $Y$ are affinely independent, and thus, every point in $C$ is a unique convex combination of vertices.  Define a map $f : C \to \D(Y)$ where $y \mapsto y $ for all $y \in Y$, and such that for every $c \in C$ that is not a vertex, rewrite $c$ as its unique convex combination of vertices. Then $c$ is of the form $r_1 \cdot y_1 + \cdots + r_n \cdot y_n$. Define the map $r_1 \cdot y_1 + \cdots + r_n \cdot y_n \mapsto \varphi$ where $\varphi(y_i) = r_i$ for $i = 1 \cdots n$. This map is clearly convex since $r_1\cdot f(y_1) + \cdots + r_n \cdot f(y_n) = r_1 \cdot \varphi_{y_1} + \cdots r_n \cdot \varphi_{y_n}$ where $\varphi_{y_i}(y_i) = 1$ is equal to $\varphi$ defined by $\varphi(y_i) = r_i$.
\end{proof}

Therefore, given a coalgebra $C \to [0,1] \times C^A$, one can find a finite set $Y$ and coalgebra $Y \to [0,1] \times \cal{D}$$(Y)^A$ in $\Set$ such that $\cal{D}$$(Y) \cong C$ (in $\Conv$) if and only if $C$ is a simplex. The set $Y$ will be precisely the vertices, or extreme points, of the simplex $C$. 

We revisit the diagram in \cref{eq:general} and obtain the following diagram in $\Conv$. 

 \begin{equation}\label{eq:diag}
 \xymatrix@C=.75cm@R=.2cm{
 X \ar[dd]_{<\out,\delta>}\ar[r]^{\eta_X}& \D (X)\ar[ddl]^-{<\out^\sharp,\delta^\sharp>}  \ar@{->>}[r]|-{\, e\, } \ar@/^{3ex}/@{-->}[rrr]|{\ \ \sem -\ \ }& I \cong \D(Y)\ar@{>->}[dr]\ar[dd] \ar@{>->}[rr]^-i&&[0,1]^{A^*}\ar[dd]_\cong\\
 &&&\vartheta F\ar@{>->}[ur]\ar[d]\\
 [0,1]\times \D(X)^A \ar@{->}[rr] &&[0,1]\times I^A\ar@{>->}[r]&[0,1]\times (\vartheta F)^A\ar@{>->}[r]& [0,1]\times  {([0,1]^{A^*})}^A \
 }
\end{equation}

Next, we apply an algorithm developed in \cite{alg} which finds the frame of the conical hull. The full algorithm can be found in \cref{sec:algo}. We also include a complete running example to show how the algorithm works in \cref{ap:ex_conv}. The algorithm remains the same for finding extreme points as for finding the conical hull, except a row of 1's is appended to the bottom of the matrix the algorithm utilizes. Therefore, this algorithm can also be used in the next subsection on WAs for the semiring $\s = \R^+$ (by not appending the row of 1's on the bottom of the matrix). Although we could use a span argument here to find the extreme points of a convex polytope, we use this algorithm for two reasons. The first reason is for efficiency, while the second reason is that we want to demonstrate any algorithm that exists that finds a base can be inserted into our approach. 

The algorithm will find the base of $\conv(X)$, the convex hull of the language acceptance vectors for $x\in X$. By \Cref{sec:minimality}, the image of $\D(X)$ in $\vartheta F$ in \cref{eq:diag} will be equivalent to $\D(\overline X)$, where $\overline X$ are the set of extreme points of $\conv(X)$.

\begin{example}[Reduced is not unique nor minimal]\label{ex:square}

It should be emphasized that the reduced automaton, as described above, is not unique unless the vertices of $\conv(X)$ are affinely independent. If they are not affinely independent, when we remove a state $x\in X \backslash \overline{X}$, we write it as one of the (possibly many) convex combination of states in $\overline{X}$. This choice influences how we define the function $f : \D(X) \to \D(\overline{X})$ above and therefore the automaton we obtain. 

Let us illustrate this through an example (see also \Cref{ex:redvsminimal}). 

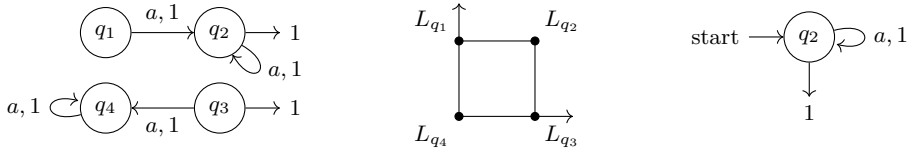
   \begin{tikzpicture}[baseline=(current bounding box.center), node distance = 1.5cm, arrows=->,font=\footnotesize]
        \tikzset{every state/.style={minimum size=-2pt}}
        \node[state] (q1) {$q_1$};
        \node[state, right of=q1, accepting  right, accepting text = {$1$}] (q2) {$q_2$};
        \node[state, below of=q2, node distance=1cm, accepting right, accepting text = {$1$}] (q3) {$q_3$};
        \node[state, below of=q1, node distance=1cm] (q4) {$q_4$};
        \draw (q1) edge[above] node{$a,1$} (q2)
              (q2) edge[out=330,in=300,looseness=8] node[right]{$a,1$} (q2)
              (q3) edge[below] node{$a,1$} (q4)
              (q4) edge[loop left] node{$a,1$} (q4);
    \end{tikzpicture}
       \hspace{1cm}
    \begin{tikzpicture}[baseline=(current bounding box.center), scale=1,font=\footnotesize]
        \coordinate (q1) at (0,1);
        \coordinate (q2) at (1,1);
        \coordinate (q3) at (1,0);
        \coordinate (q4) at (0,0);
        
        \draw[fill=black] (q1) circle (0.15em)
        node[above left] {$L_{q_1}$};
        \draw[fill=black] (q2) circle (0.15em)
        node[above right] {$L_{q_2}$};
        \draw[fill=black] (q3) circle (0.15em)
        node[below right] {$L_{q_3}$};
        \draw[fill=black] (q4) circle (0.15em)
        node[below left] {$L_{q_4}$};
        
        \draw[->] (0,0) -- (1.5, 0);
        \draw[->] (0,0) -- (0, 1.5) node[above] {\ };
        \draw[-] (0,1) -- (1,1);
        \draw[-] (1,1) -- (1,0);
    \end{tikzpicture}
     \hspace{1cm}
      \begin{tikzpicture}[arrows=->,baseline=(current bounding box.center),font=\footnotesize]
        \tikzset{every state/.style={minimum size=-2pt}}
        \node[initial, state, accepting below, accepting text = {$1$}] (q1) {$q_2$};
        \draw (q1) edge[loop right] node{$a,1$} (q1);
    \end{tikzpicture}

The automaton on the left has state $q_1$ whose language $L_{q_1}$ can be represented by the vector $(0,1)$ in $\mathbb{R}^2$ (depicted in the middle), similarly we have $q_2$ as (1,1), $q_3$ as (1,0), and $q_4$ as (0,0). The convex hull of these points gives a square in $\mathbb{R}^2$, and all states are extreme points (and hence the automaton is reduced). It is immediately clear these points are not affinely independent, since $\mathbb{R}^n$ can have at most $n+1$ affinely independent points. Take for instance the center of the square $(\frac{1}{2}, \frac{1}{2})$. This can be written as both $\frac{1}{2} q_1 + \frac{1}{2} q_3$ and also as $\frac{1}{2} q_2 + \frac{1}{2} q_4$. Once a point can be written as two different convex combinations, it can be written as infinitely many different convex combinations. Thus, the choice $f$ makes will determine what convex combination of states are used when eliminating a specific point. 

This example also illustrates the distinction between reduction and minimization. If we choose initial state $q_2$, a minimal automaton recognizing the language $L_{q_2}$ is depicted on the right above. The key difference between the reduced and the minimal automaton is that in the latter {\em reachability} has been taken into account, whereas the reduced automaton contains all states necessary for {\em any} choice of initial distribution. 
\end{example}

\subsubsection{Reducing using probabilistic observation tables}\label{datastructure}

We now present an algorithm for reducing probabilistic automata, providing a data structure that can be used to distinguish the different languages accepted by the states of the original automaton. We fix a finite probabilistic automaton $(Q,<\out,\delta>)$.
The data structure we use to reduce this automaton stems from Angluin's learning algorithm for DFAs~\cite{angluin}.
Although our rows are labeled by states rather than words, we reuse the name \emph{observation table} here. For readability we will use $\obs$ to denote the final coalgebra map $\sem - \colon \D(X) \to [0,1]^{A^*}$. For $x\in X$ we abuse notation and write $\obs(x)$ instead of $\obs(\eta(q))$.

\begin{definition}[Observation table]
    An \emph{observation table} is defined by a set of words $E \subseteq A^*$.
    It is the function $\row_E \colon Q \to [0, 1]^E$ defined by $\row_E(q)(e) = \obs(q)(e) $.
    We often identify the table by the set $E$.
\end{definition}

\begin{example}\label{fig:tetra}
Consider the automaton on the left below. On the right is an observation table for  $E=\{\varepsilon, a, aa\}$.
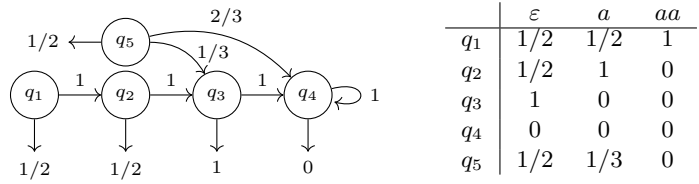
\begin{center}
\vspace*{-.2cm}
    \begin{tikzpicture}[baseline=(current bounding box.center), node distance = 1.2cm, arrows=->,font=\scriptsize,scale=.95]
        \tikzset{every state/.style={minimum size=0pt}}
        \node[state, accepting below, accepting text = {$1/2$}] (q1) {$q_1$};
        \node[state, right of=q1, accepting below, accepting text = {$1/2$}] (q2) {$q_2$};
        \node[state, right of=q2, accepting below, accepting text = {$1$}] (q3) {$q_3$};
        \node[state, right of=q3, accepting below, accepting text = {$0$}] (q4) {$q_4$};
        \node[state, above of=q2, accepting left, node distance=.7cm, accepting text = {$1/2$}] (q5) {$q_5$};
        \draw (q1) edge[above] node{$1$} (q2)
              (q2) edge[above] node{$1$} (q3)
              (q3) edge[above] node{$1$} (q4)
              (q4) edge[loop right] node{$1$} (q4)
              (q5) edge[bend left,right,pos=.6] node{$1/3$} (q3)
              (q5) edge[bend left, above] node{$2/3$} (q4);
    \end{tikzpicture}
    \hspace{0.5cm}
   {\footnotesize \begin{tabular}{ c | c c c }
        & $\eword$ & $a$ & $aa$ \\
        \hline
        $q_1$ & $1/2$ & $1/2$ & $1$ \\[.5ex]
        $q_2$ & $1/2$ & $1$ & $0$ \\[.5ex]
        $q_3$ & $1$ & $0$ & $0$ \\[.5ex]
        $q_4$ & $0$ & $0$ & $0$ \\[.5ex]
        $q_5$ & $1/2$ & $1/3$ & $0$
    \end{tabular}}
    \end{center}

    \end{example}
The image of the $\row_{A^*}$ function clearly generates the image of $\obs$, and from this finite set of generators we want to extract the base which will be used as states of the reduced automaton. 
However, this table has an infinite number of columns.
It is often useful to consider convex combinations of the rows of an observation table. As such, we introduce the following definition.

\begin{definition}[Probabilistic extension of $\row$]
For each $E \subseteq A^*$, we define the {\em probabilistic extension of $\row$}, written $\overline{\row}_E \colon \D(Q) \to [0, 1]^{E}$, to be the unique homomorphism extending $\row_E$. For $v_1,v_2 \in \D(Q)$, we write $v_1 \equiv_E v_2$ when $\overline{\row}_E(v_1)(e) = \overline{\row}_E(v_2)(e)$ for all $e \in E$.
\end{definition}

In particular, we have $\overline{\row}_{A^*} = \obs$. Observe that for each $E \subseteq A^*$, the image of $\overline{\row}_E$ is a convex polytope $P_E \subseteq [0,1]^E$. When $E \subseteq F \subseteq A^*$, the restriction map from $[0,1]^F$ to $[0,1]^E$ yields a surjective homomorphism from $P_F$ to $P_E$. To determine which states in $Q$ correspond to extreme points of $P_{A^*}$, we wish to find a finite set $E$ such that the restriction map from $P_{A^*}$ to $P_E$ is an isomorphism. 

\begin{definition}
Let $E \subseteq A^*$. We inductively define, for each $k \in \N$ 
$$A^0E \eqdef E\quad\text{ and }\quad A^{k+1}E \eqdef \{aw \mid a \in A, w \in A^kE \}.$$ 
We then define $A^*E \eqdef \bigcup_{k \in \N} A^kE = \{we \mid w \in A^*, e \in E\}$.
\end{definition}

\begin{definition}[Consistency]
Let $E \subseteq A^*$. We say that the observation table corresponding to $E$ is \emph{consistent} if for all $v_1, v_2 \in \D(Q)$, we have $v_1 \equiv_E v_2$ implies $v_1 \equiv_{AE} v_2$.
\end{definition}

It can be shown that when $E$ is consistent, the equivalence relations $\equiv_E$ and $\equiv_{A^*E}$ coincide. We focus primarily on the case where $\eword \in E$.

\begin{theorem}
If $E \subseteq A^*$ is consistent and $\eword \in E$, then for all $v_1, v_2 \in \D(Q)$ we have $v_1 \equiv_E v_2$ implies $\obs(v_1) = \obs(v_2)$.
\end{theorem}

\begin{proof}
Suppose $v_1 \equiv_E v_2$ and $w \in A^*$ has length $k \in \N$. Then $w = w\eword \in A^kE$. Since $v_1 \equiv_{A^kE} v_2$ (\Cref{lemma:consistent}), we conclude $\obs(v_1)(w) = \overline{\row}_{A^kE}(v_1)(w) = \overline{\row}_{A^kE}(v_1)(w) = \obs(v_2)(w) $. 
\end{proof}

\begin{corollary}
If $E \subseteq A^*$ is consistent and $\eword \in E$, then the restriction map from the image of $\obs$ to the image of $\overline{\row}_E$ is an isomorphism.
\end{corollary}

Using this data structure, we now shift our attention to determining whether or not an observation table with a finite number of columns is consistent. The main result is summarized in the following theorem.
\begin{theorem}\label{thm:pol}
There is a polynomial-time algorithm which determines whether or not the observation table $E = \{e_1,e_2,\dots,e_m\}$ with $m \in \N$ is consistent. When $E$ is inconsistent, the algorithm produces a pair $(a,i)$ with $a \in A$ and $1 \leq i \leq m$ such that $\obs(v_1)(ae_i) \neq \obs(v_2)(ae_i)$ for some $v_1,v_2 \in \D(Q)$ satisfying $v_1 \equiv_E v_2$.
\end{theorem}

In \cref{sec:redpa}, we include further details on how to build a consistent table and a reduced automaton from it.  Although we currently only have the above result for probabilistic automata, in the future, it would be interesting to study how far it generalizes for other functors and monads, as well. 

\subsection{$\s$-linear Weighted Automata}\label{sec:wa}

Our second example of finitary automata is $\s$-linear weighted automata (WA).

\begin{definition}[Free Semimodule Monad]\label{semi}

Let $\Se_\s : \Set \to \Set$ denote finitely supported $\s$-combinations on a set $X$, defined as 
$\Se_\s(X) = \left\{ \phi:X \to \s \mid \sum \phi(x)\in \s \right\}$ where the support of each $\phi$ is finite.

\end{definition}

\begin{definition}[$\s$-linear weighted automaton (WA)]
    A \emph{$\s$-linear weighted automaton} over a semiring $\s$ is a tuple $(Q, <\delta, \out>)$, consisting of a set of states $Q$, a transition function $\delta \colon Q \to \Se_\s(Q)^A$, and an output function $\out \colon Q \to \s$.
\end{definition}

An $\s$-linear weighted automaton can be viewed as the coalgebra $F(X)=\s \times \Se_\s(X)^A$ on $\EM(\Se_\s)$, the category of $\s$-semimodules, where $\Se_\s(\mathbf{1})=\s$. Note $\Se_\s$ is a finitary monad. As we alluded to earlier, our algorithm does not currently work for all semirings $\s$, since finding a base of certain semirings is not always feasible. We explain this in further detail below. The goal is to reduce the number of states by finding a minimal subset $\overline{X}$ such that $span(X) = span(\overline{X})$ in $\SMod(\s)$, the span (over $\s$) of the language vectors for all $x\in X$. 

WAs differ from PAs, where the convex hull has a unique base (its extreme points). Finitely generated $\s$-semimodules do not have this property. First, they do not necessarily have a unique base, thus even if we can find a reduced automaton, it may not be unique in its set of states; i.e. there may exist two different reduced automata with different state spaces. This part is not a problem assuming all bases have the same cardinality, because if this is the case, different reduced automata would be isomorphic. The problem we can encounter is that there may exist different bases of an $\s$-semimodule of different cardinality. Therefore, there are some semirings where it may not currently be feasible to find a reduced automaton (with any known algorithms), even though we may still be able to find an automaton with a smaller state space.

 \paragraph*{Reverse determinization: from $\SMod$ to $\Set$.}

\begin{theorem}
An $\s$-semimodule $M$ is isomorphic to $\Se_\s(Y)$ as semimodules for a finite set $Y$ iff $M\cong \s^n$ for some $n\in \mathbb{N}$. 
\end{theorem}

\begin{proof}
$\Se_\s(Y)$ is a free and finitely generated $\s$-semimodule, and all free and finitely generated semimodules are of the form $\s^n$. 
\end{proof}

Therefore, given a coalgebra $S \to \s \times S^A$ for an $\s$-semimodule $S$, one can find a finite set $X$ and coalgebra $X \to \s \times \Se_s(X)^A$ in $\Set$ such that $\Se_\s(X) \cong S$ (in $\SMod$) if and only if $S \cong \s^n$ for some $n \in \N$.

\begin{remark}\label{remark:conditions} We must assume certain conditions on the semiring $\s$ in order to reduce. The first is we assume $\s$ is a commutative semiring. The second assumption is $k(\s)=1$, where $k(\s)=max\{t \in \mathbb{N} \mid$ the $\s$-semimodule $\s$ has a basis with t elements$\}$ (\cite{base}). Examples of semirings satisfying these conditions include vector spaces, Noetherian local semirings, $R_+$ for a ring $R\subset \mathbb{R}$ and skew fields.
\end{remark}

We revisit the diagram in \cref{eq:general} and obtain the following diagram in $\SMod$.

 \begin{equation}\label{diag}
 \xymatrix@C=.75cm@R=.2cm{
 X \ar[dd]_{<\out,\delta>}\ar[r]^{\eta_X}& \Se_\s (X)\ar[ddl]^-{<\out^\sharp,\delta^\sharp>}  \ar@{->>}[r]|-{\, e\, } \ar@/^{3ex}/@{-->}[rrr]|{\ \ \sem -\ \ }& I \cong \Se_\s(Y)\ar@{>->}[dr]\ar[dd] \ar@{>->}[rr]^-i&&\s^{A^*}\ar[dd]_\cong\\
 &&&\vartheta F\ar@{>->}[ur]\ar[d]\\
 \s\times \Se_\s(X)^A \ar@{->}[rr] &&\s\times I^A\ar@{>->}[r]&\s\times (\vartheta F)^A\ar@{>->}[r]& \s\times  {(\s^{A^*})}^A \
 }
\end{equation}

To relate $\Se_\s(X)$ with its reduced state space, we need to find a base of $\Se_\s(X)$. We utilize a span approach as follows.

Choose any $x\in X$ and check if $x\in \mathbf{span}(X \backslash \{x\})$. If $x$ is in the span, remove $x$ from $X$. If $x$ is not in the span, keep $x$ in $X$. Continue in this way until we do this for all $x\in X$. Since we assume $\s$ is commutative and $k(\s)=1$, all bases have the same cardinality and thus the order in which we check each $x\in X$ does not matter. This algorithm will find a base of $\mathbf{span}(X)$, the span of the language acceptance vectors for $x\in X$. By \Cref{sec:minimality}, the image of $\D(X)$ in $\vartheta F$ in \cref{eq:diag} will be equivalent to $\Se_\s(\overline X)$, where $\overline X$ is the base.

\begin{remark}If a semiring $\s$ does not satisfy commutativity and $k(\s) \neq 1$, there is no guarantee we will obtain a reduced automata using this approach. We may still be able to obtain a smaller automata, but it may not be a reduced one. Recall from the last section that we did not utilize a span approach to find a base for probabilistic automata, so there may exist other algorithms for a semiring $\s$ where $k(\s) \ge1$ that can find a minimal size base. Inserting such an algorithm would allow for reduction of the WA.
\end{remark}

\section{Discussion and Future Work}\label{sec:future}

In this paper, we studied {\em reduced finitary automata}, which are finitary automata with no redundant states. Two well-known examples of finitary automata are probabilistic automata and $\s$-linear weighted automata for a large class of semirings $\s$. Reduced automata are not necessarily minimal (see e.g. \Cref{ex:redvsminimal,ex:square}) as the reduction process does not take into account a particular initial state. Because of this, we end up with an automaton that is not minimal for a particular initial state but instead can then be used to look at different initial states -- this of course comes at the cost of having extra states, yet with no redundancy. 

Reduced finitary automata fit naturally in the coalgebraic framework, where it is common to ignore the initial state. The coalgebraic outlook gave us valuable guidance on the algebraic structures involved in the reduction process, justifying when we could go back to a set-based automaton. We strongly believe this approach will guide us in other generalizations, e.g to work with coalgebras on $\EM(\M)$, where $\M$ is not over $\Set$.

Our work also opens a door to a potential relationship between reduction and properness. Recall that a weighted automaton for a semiring $\s$ is reducible using the span approach discussed in \cref{sec:wa} iff $k(\s)=1$. We conjecture that if weighted automata over the semiring $\s$ are reducible (using the span approach), then $\s$ is a proper semiring. From \cite{proper}, to prove properness of some cubic functors, the authors create a zig-zag with the middle node $Z=<(c_{1_w}(x), c_{2_w}(y)>_{w\in A^*}$, a subsemimodule of $\s^{n_1} \times \s^{n_2}$, where $x\in \s^{n_1}$ and $y\in \s^{n_2}$ are assumed to be trace equivalent. If $Z$ is always finitely generated for any $x$ and $y$ trace equivalent, it then follows that the semiring $\s$ is proper. We would like to show for reducible semirings $\s$, $Z$ is finitely generated.


We envisage that the above problem can be tackled as follows. The accepted language of a WA over $\s$ can be viewed as a subset of $\s << A^*>>$, the set of formal power series. In \cite{residual}, there is a connection between residual languages and finitely generated subsemimodules of $\s<<A^*>>$, although only the semirings $\mathbb{R}, \mathbb{Q}, \mathbb{R}_+,$ and $\mathbb{Q}_+$ are considered; however, we note these four semirings are known to be proper. If the subsemimodule of $\s<<A^*>>$ representing the accepted language of a WA is finitely generated, then only a finite amount of words $\{w_i\}_i \subsetneq A^*$ need to be checked to understand how the automata behaves on all of $A^*$. If this holds, this would then show for $\s$-linear weighted automata $\mathcal{A}$ and $\mathcal{B}$ where states $x\in \mathcal{A}$ and $y\in \mathcal{B}$ are trace equivalent, $Z$ is generated by $m_1 \cdot m_2$ elements, where $m_1$ is the number of words that need to be checked for $\mathcal{A}$ and $m_2$ is the number of words that need to be checked for $\mathcal{B}$. Using results from   \cite{residual,shortwordpa} we want to try to show for any $\s$-linear weighted automata with $k(\s)=1$, the subsemimodule of $\s<<A^*>>$ representing the accepted language of the automaton is finitely generated. If this can be done, this would show that if  weighted automata over $\s$ are reducible, then $\s$ is proper.


Furthermore, the idea of proper semirings was extended to proper functors in \cite{fixpoint}. In \cite{proper}, it was proven that the functor $F(X)=[0,1]\times X^A$ is a proper functor on $\mathbf{PCA}$. This leads us to one of our foremost goals: to determine whether the functor $F(X)=[0,1] \times X^A$ on $\Conv$ is proper. We have evidence to support this is true, but we leave this as future work.

\bibliography{bib}

\appendix

\section{Algorithm to compute extreme points}\label{sec:algo}

The following algorithm from \cite{alg} begins with a matrix $A$ with columns as points in $\mathbb{R}^m$. It determines the frame of the conical hull of these points. The matrix $A$ is inputted into this algorithm, and the algorithm will produce a matrix that will contain columns which are a subset of the columns of $A$, which give the frame of the conical hull.

\medskip
\noindent\textbf{Algorithm.}
\begin{enumerate}[i]
\item (Canonical form) Use Jordan elimination for bringing $A$ into canonical form. Every column is labeled undecided.
\item (Constant column) Select a nonbasic undecided column $A_c$ as "constant column." If there are no such columns, go to (viii); else proceed to (iii).
\item (Pilot row) If $A_c \geq 0$, delete $A_c$ ad return to (ii). Else select $_pA$ such that $a_{pc} \lt 0$; call it the "pilot row."
\item (Pivot column) If $a_{pc}$ is the only negative entry in the pilot row, change the label of $A_c$ from "undecided" to "necessary" and go to (ii). Else select a "pivot column" $A_l$ such that $a_{pl} \lt 0$ and $l\neq c$.
\item (Pivot row) Select a "pivot row" $_rA$ such that
$$ 0 \leq \frac{a_{rc}}{a_{rl}} \leq \min\left\{\frac{a_{ic}}{a_{il}} | a_{ic} \geq 0, a_{il} \gt 0 \right\}$$
\item (Pivoting) Execute a simplex step (=Jordan transformation) with $a_{rl}$ as pivot. In the absence of degeneracies, this will increase the old entry $a_{pc}$ while keeping nonnegative entries of $A_c$ nonnegative.
\item (Return) If the new entry $a_{pc}$ is still negative, keep the $p^{th}$ row as pilot row and go to (iv). Else go to (iii).
\item (Termination) If all basic columns are decided, terminate the procedure. Else select an undecided basic column $A_c$.
\item (Clear basis) Suppose $a_{rc} = 1$. If this is the only positive entry in $_rA$, then change the label of the column $A_c$ from "undecided" to "necessary" and go to (viii). Else pivot so as to remove the undecided column $A_c$ from the basis. Go to (iii) with $A_c$ as constant column. (End of the algorithm).
\end{enumerate}

\subsection{Example Application}\label{ap:ex_conv}

Take the points $\hat{A}_1=(2,2), \hat{A}_2=(4,2), \hat{A}_3=(2,4), \hat{A}_4=(4,4), \hat{A}_5=(3,3), \hat{A}_6=(4,3)$.

We perform the algorithm above on the matrix

\begin{center}
$A=\begin{bmatrix}
2 & 4 & 2 & 4 & 3 & 4\\
2 & 2 & 4 & 4 & 3 & 3\\
1 & 1 & 1 & 1 & 1 & 1
\end{bmatrix}$
\end{center}

(i) Jordan elimination brings $A$ into canonical form.

\begin{center}
$A=\begin{bmatrix}
1 & 0 & 0 & -1 & 0 & -1/2\\
0 & 1 & 0 & 1 & 1/2 & 1\\
0 & 0 & 1 & 1 & 1/2 & 1/2
\end{bmatrix}$
\end{center}

(ii) Choose $A_4$ non-basic undecided column. Call $A_4$ the "constant column".

(iii) $A_4$ is not non-negative, thus select row $_1A$ since $a_{14} \lt 0$ and call $_1A$ the "pivot row".

(iv) $a_{14}$ is not the only negative entry in $_1A$ since $a_{16} \lt 0$. Thus, $A_6$ is a "pivot column".

(v) $a_{24}, a_{34} \geq 0$ and $a_{26}, a_{36} \gt 0$. The minimum of the set is $1$, so select $_2A$ to be the "pivot row".

(vi) $a_{26}$ is the pivot point. We execute a Jordan elimination step with $a_{26}$ as the pivot. This gives

\begin{center}
$A=\begin{bmatrix}
1 & 1/2 & 0 & -1/2 & 1/4 & 0\\
0 & 1 & 0 & 1 & 1/2 & 1\\
0 & 1/2 & 1 & 1/2 & 1/4 & 0
\end{bmatrix}$
\end{center}

(vii) $a_{14}$ is still negative so keep the first row $_1A$ as the pilot row and go to (iv).

(iv) $a_{14}$ is the only negative entry in the pilot row so change $A_4$ from undecided to necessary. Thus, $(4,4)$ is a vertex. Go to (ii).

(ii) Select $A_5$ as "constant column". $A_5 \geq 0$ so delete $A_5$.

\begin{center}
$A=\begin{bmatrix}
1 & 0 & 0 & -1 & -1/2\\
0 & 1 & 0 & 1 & 1\\
0 & 0 & 1 & 1 & 1/2
\end{bmatrix}$
\end{center}

Go to (ii).

(ii). Select $A_6$. 

(iii) $A_6$ is not non-negative, so select $_1A$ since $a_{16} \lt 0$ and $_1A$ become the "pilot row". $a_{16}$ is not the only non-negative entry in $_1A$ since $a_{14} \lt 0$. Select "pivot column" $A_4$. 

(v) The minimum of the set, where $l=4, p=1,$ and $c=6$, is $1/2$. Thus, choose "pivot row" $_3A$. 

(vi) Executing a simplex step with $a_{34}$ as pivot gives

\begin{center}
$A=\begin{bmatrix}
1 & 0 & 1 & 0 & 0\\
0 & 1 & -1 & 0 & 1/2\\
0 & 0 & 1 & 1 & 1/2
\end{bmatrix}$
\end{center}

(vii) $a_{16}$ is no longer negative. Go to (iii).

(iii) $A_6 \geq 0$ so delete $A_6$. 

\begin{center}
$A=\begin{bmatrix}
1 & 0 & 0 & -1\\
0 & 1 & 0 & 1\\
0 & 0 & 1 & 1
\end{bmatrix}$
\end{center}

Go to (ii).

(ii) No more undecided nonbasis columns. Go to (viii).

(viii) Select $A_1$. 

(ix) $a_{11}=1$ and is the only positive entry in row 1. Thus, $A_1$ is necessary and $(2,2)$ is a vertex.

(viii) Select $A_2$. 

(ix) $a_{22}$ is not the only positive entry in row 2. Pivot as to remove the undecided column $A_2$ from the basis. 

\begin{center}
$A=\begin{bmatrix}
1 & 0 & 0 & 1\\
0 & 1 & 0 & -1\\
0 & 0 & 1 & 1
\end{bmatrix}$
\end{center}

where we moved $A_2$ to the last column to put the basis elements together.

(ii) $A_2$ is a "constant column".

(iii) $A_2$ is not non-negative so select $_2A$ since $a_{24} \lt 0$ as the "pilot row".

(iv) $a_{24}$ is the only negative entry in row 2 so $A_2$ becomes necessary, and $(4,2)$ is a vertex.

(viii) Choose $A_3$.

(ix) $a_{33}=1$, but is not the only positive entry in row 3 because $a_{34} \gt 0$.  Pivot as to remove the undecided column $A_3$ from the basis.

\begin{center}
$A=\begin{bmatrix}
1 & 0 & 0 & 1\\
0 & 1 & 0 & -1\\
0 & 0 & 1 & 1
\end{bmatrix}$
\end{center}

Go to (iii)

(iii) $A_3$ is not non-negative so select $_2A$ since $a_{24} \lt 0$ as the "pilot row".

(iv) $a_{24}$ is the only negative entry in row 2, thus, $A_3$ becomes necessary, and $(2,4)$ is a vertex.

\section{Extra \Cref{datastructure}}

\begin{lemma}
If $E \subseteq A^*$ is consistent, then $AE$ is consistent.
\end{lemma}

\begin{proof}
Let $\overline{\delta} \colon \D(Q) \to \D(Q)^A$ be the unique homomorphism extending $\delta$. Suppose $v_1, v_2 \in \D(Q)$ with $v_1 \equiv_{AE} v_2$. For each $a \in A$ and $e \in E$, we have
\[ \overline{\row}_E(\overline{\delta}(v_1)(a))(e) = \overline{\row}_{AE}(v_1)(ae) = \overline{\row}_{AE}(v_2)(ae) = \overline{\row}_E(\overline{\delta}(v_2)(a))(e), \]
hence $\overline{\delta}(v_1)(a) \equiv_{E} \overline{\delta}(v_2)(a)$. Then consistency of $E$ yields $\overline{\delta}(v_1)(a) \equiv_{AE} \overline{\delta}(v_2)(a)$ for each $a \in A$. In particular, for $w \in AE$, we have
\[ \overline{\row}_{A^2E}(v_1)(aw) = \overline{\row}_{AE}(\overline{\delta}(v_1)(a))(w) = \overline{\row}_{AE}(\overline{\delta}(v_2)(a))(w) = \overline{\row}_{A^2E}(v_2)(aw), \]
so $v_1 \equiv_{A^2E} v_2$ as needed.
\end{proof}

\begin{lemma}\label{lemma:consistent}
If $E \subseteq A^*$ is consistent and $v_1,v_2 \in \D(X)$ satisfy $v_1 \equiv_E v_2$, then for each $k \in \N$ we have $v_1 \equiv_{A^kE} v_2$.
\end{lemma}

\begin{proof}
Straightforward induction on $k$ using the previous lemma.
\end{proof}

\subsection{Reduction algorithm for probabilistic automata}\label{sec:redpa}
It is convenient for computational purposes to represent observation tables and transition functions as real-valued matrices.

\begin{definition}\label{def:matrices}
When $Q = \{q_1,q_2,\dots,q_n\}$ and $E = \{e_1,e_2,\dots,e_m\}$ with $n,m \in \N$, we may view the observation table corresponding to $E$ as the real-valued $n$-by-$m$ matrix $M_E \in [0,1]^{n \times m}$ given by $(M_E)_{i,j} \eqdef \row_E(q_i)(e_j)$ for $1 \leq i \leq n$ and $1 \leq j \leq m$. Similarly, we may view $\delta$ as the family of matrices $\{ D_a \}_{a \in A}$ such that each $D_a \in [0,1]^{n \times n}$ is given by $(D_a)_{i,j} \eqdef \delta(q_i)(a)(q_j)$ for $1 \leq i,j \leq n$.
\end{definition}

\begin{definition}
For each $n \in \mathbb{N}$, we say that a vector $t \in [0,1]^n$ is \emph{stochastic} when $\sum_{i=1}^n t_i = 1$. When $|Q| = n$, the set of stochastic vectors in $[0,1]^n$ is isomorphic to $\D(Q)$.
\end{definition}

Assume $|Q| = n$ and $|E| = m$ with $n,m \in \N$. Observe that the image of $\overline{\row}_E$ coincides with the set of row vectors $\{ t^T M_E \mid t \in [0,1]^n \ \text{is stochastic} \}$. Furthermore, observe that for each $a \in A$, the matrix product $D_a M_E$ is the observation table for the set of words $\{ae \mid e \in E\}$. It follows that $E$ is inconsistent precisely when there exists $a \in A$ and stochastic vectors $t,s \in [0,1]^n$ such that $t^T M_E = s^T M_E$ but $t^T D_a M_E \neq s^T D_a M_E$. This motivates the following algorithm, which takes as input the following arguments
\begin{itemize}
    \item An observation table $M \in \R^{n \times m}$
    \item A transition matrix $D \in \R^{n \times n}$
\end{itemize}
and decides whether or not the pair $(M,D)$ is consistent; that is, whether or not for all stochastic vectors $t,s \in [0,1]^n$, we have \begin{equation}\label{eq:cons} t^T M = s^T M  \Rightarrow t^T D M = s^T D M. \end{equation}

\begin{remark}[Null space] If $t$ and $s$ were not required to be stochastic, consistency would merely amount to $\text{null}(M^T) \subseteq \text{null}(M^TD^T)$. This can be decided in polynomial time by using Gaussian elimination to find a basis for $\text{null}(M^T)$. \end{remark}

\noindent{\bf Consistency check for $(M,D)$.} To solve \Cref{eq:cons} we use linear programming, and determine whether or not there exist $t,s \in \R^n$ satisfying the following linear constraints:
$$t_i, s_i \geq 0\text{ for }1 \leq i \leq n \qquad\quad
    \sum_{i = 1}^n t_i = 1 = \sum_{i = 1}^n s_i \qquad\quad
 \sum_{i=1}^n (t_i - s_i) M_{i,j} = 0\text{ for }1 \leq j \leq m$$
and such that $t^T D M \neq s^T D M$. The first two sets of constraints state that $t$ and $s$ are stochastic. The third constraint states that $(t - s)^T M = 0$. Note that $t^T D M \neq s^T D M$ precisely when, for some $1 \leq j \leq m$, the $j$'th component of $(t - s)^T D M$ is nonzero. Since $t$ and $s$ are interchangeable, we may assume this component is positive. For fixed $j$, this is captured by the linear expression
\[ \sum_{i = 1}^n (t_i - s_i) (D M)_{i,j} \gt 0. \]
For each $j$, we can maximize the objective function $\sum_{i = 1}^n (t_i - s_i) (D M)_{i,j}$ subject to the above linear constraints by solving the corresponding linear program. By construction, $(M,D)$ is inconsistent iff at least one of the $m$ objective functions attains a positive value. It is clear that each program can be constructed in polynomial time. Since linear programs can be solved in polynomial time (depending on the algorithm chosen this will be in $O(n^4L)$), the consistency of $(M,D)$ can be decided in polynomial time.

\medskip

\noindent{\bf Building a consistent table.} We now can build the reduced automaton by incrementally constructing $E$ and then checking if the corresponding observation table is consistent. This leads to iteratively discovering the (state space) polytope dimension. If we consider the automaton from \Cref{fig:tetra} we end up with 3 tables for $E=\{\varepsilon\}$, $E=\{\varepsilon, a\}$ and $E=\{\varepsilon, a, aa\}$ (see \Cref{fig:tetra} for the last table, which is {\em consistent},  the other two can be read as subtables thereof by restricting the columns), which can be depicted as follows: 
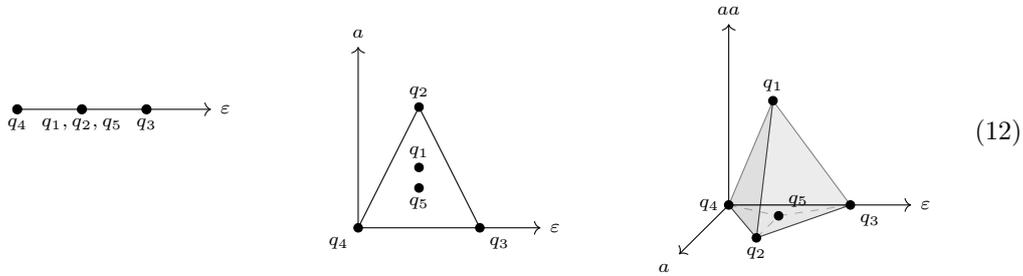
\begin{equation}\label{eq:tetra}
 \begin{tikzpicture}[baseline=(current bounding box.south east), scale=1.7,font=\scriptsize]
        \coordinate (left) at (0,0);
        \coordinate (mid) at (0.5,0);
        \coordinate (right) at (1,0);
        
        \draw[fill=black] (left) circle (0.1em)
            node[below] {$q_4$};
        \draw[fill=black] (mid) circle (0.1em)
            node[below] {$q_1,q_2,q_5$};
        \draw[fill=black] (right) circle (0.1em)
            node[below] {$q_3$};
            
        \draw[->] (0,0) -- (1.5,0) node[right] {$\eword$};
    \end{tikzpicture}
     \qquad\quad
    \begin{tikzpicture}[baseline=(current bounding box.center), scale=1.6,font=\scriptsize]
        \coordinate (q1) at (0.5,0.5);
        \coordinate (q2) at (0.5,1);
        \coordinate (q3) at (1,0);
        \coordinate (q4) at (0,0);
        \coordinate (q5) at (0.5,0.33);
        
        \draw[fill=black] (q1) circle (0.1em)
            node[above] {$q_1$};
        \draw[fill=black] (q2) circle (0.1em)
            node[above] {$q_2$};
        \draw[fill=black] (q3) circle (0.1em)
            node[below right] {$q_3$};
        \draw[fill=black] (q4) circle (0.1em)
            node[below left] {$q_4$};
        \draw[fill=black] (q5) circle (0.1em)
            node[below] {$q_5$};
            
        \draw[->] (0,0) -- (1.5,0) node[right] {$\eword$};
        \draw[->] (0,0) -- (0,1.5) node[above] {$a$};
        \draw[-] (q4) -- (q2);
        \draw[-] (q2) -- (q3);
    \end{tikzpicture}
    \qquad\quad
    \begin{tikzpicture}[baseline=(current bounding box.center), scale=1.6,font=\scriptsize]
        \pgfmathsetmacro{\factor}{1/sqrt(2)};
        \coordinate (q1) at (0.5,1,0.5*\factor);
        \coordinate (q2) at (0.5,0,1*\factor);
        \coordinate (q3) at (1,0,0);
        \coordinate (q4) at (0,0,0);
        \coordinate (q5) at (0.5,0,0.33*\factor);
        
        \foreach \i in {2,3,4}
            \draw[dashed] (q\i)--(q5);
        
        \draw[-, fill=gray!10, opacity=.5] (q2)--(q3)--(q4)--cycle;
        \draw[-, fill=gray!30, opacity=.5] (q1)--(q2)--(q3)--cycle;
        \draw[-, fill=gray!50, opacity=.5] (q1)--(q2)--(q4)--cycle;
        
        \draw[fill=black] (q1) circle (0.1em)
            node[above] {$q_1$};
        \draw[fill=black] (q2) circle (0.1em)
            node[below] {$q_2$};
        \draw[fill=black] (q3) circle (0.1em)
            node[below right] {$q_3$};
        \draw[fill=black] (q4) circle (0.1em)
            node[left] {$q_4$};
        \draw[fill=black] (q5) circle (0.1em)
            node[above right] {$q_5$};
        
        \draw[->] (0,0) -- (1.5,0,0) node[right] {$\eword$};
        \draw[->] (0,0) -- (0,1.5,0) node[above] {$aa$};
        \draw[->] (0,0) -- (0,0,1.5*\factor) node[below left] {$a$};
    \end{tikzpicture}
    \end{equation}
   
\noindent{\bf Building the reduced automaton from a consistent table.} 
Given the extreme points $r_1, r_2, \ldots, r_n \in [0, 1]^E$ of the convex combinations of the rows of a consistent table, which we can calculate using the algorithm of 
\Cref{sec:algo}, we know these extreme points are included in the rows of the table, since the convex algebra is generated from those rows.

Thus, there are $q_1, q_2, \ldots q_n \in Q$ such that $\mathsf{row}(q_i) = r_i$ for $1 \le i \le n$.
The reduced automaton can now be built as follows:
Let $Q' = \{q_1, q_2, \ldots, q_n\}$ be its state space.
For any state $q \in Q$ we can determine a convex combination of states in $\D(Q')$ accepting the same language as $q$, by looking at the rows corresponding to $q$ and the states in $Q'$. This allows us to restrict the transition functions from $Q$ to $Q'$.
The output function can also be restricted to $Q'$.

For example, we can see that the table corresponding to the tetrahedron in \Cref{eq:tetra} (see also \Cref{fig:tetra}) will be consistent and from that we can obtain the reduced automaton with states $\{q_1, q_2, q_3, q_4\}$ (by computing the extreme points using the algorithm from \Cref{sec:algo}). The transition structure will be obtained by removing $q_5$ and it would replace its incoming  transitions with the appropriate convex combination (the row for $q_5$ is a convex combination of the rows for $q_2$, $q_3$, and $q_4$ all weighted by $\frac 1 3$). Since $q_5$ does not have incoming transitions the resulting automaton is: 

\smallskip
   
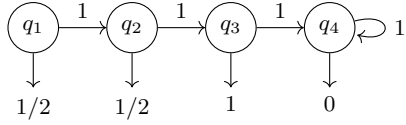
\begin{tikzpicture}[baseline=(current bounding box.center), node distance = 1.3cm, arrows=->,font=\footnotesize]
        \tikzset{every state/.style={minimum size=0pt}}
        \node[state, accepting below, accepting text = {$1/2$}] (q1) {$q_1$};
        \node[state, right of=q1, accepting below, accepting text = {$1/2$}] (q2) {$q_2$};
        \node[state, right of=q2, accepting below, accepting text = {$1$}] (q3) {$q_3$};
        \node[state, right of=q3, accepting below, accepting text = {$0$}] (q4) {$q_4$};
        \draw (q1) edge[above] node{$1$} (q2)
              (q2) edge[above] node{$1$} (q3)
              (q3) edge[above] node{$1$} (q4)
              (q4) edge[loop right] node{$1$} (q4);
    \end{tikzpicture}

\smallskip

   \medskip
\noindent{\bf Note on termination and correctness-} The above iterative process will terminate because at each step we are increasing the size of the potential basis: in the worst case scenario, the automaton we started with is already reduced and we will discover that the extreme points include all the states. The guarantee of progression towards the correct observation table is captured by the following theorem: 
\begin{theorem}[Quotient Progress Property]\label{thm:progress} Assume there exists a surjective convex map between convex polytopes $f : C_1 \twoheadrightarrow C_2$ that is not an isomorphism. Then $C_2$ has strictly lower dimension than $C_1$.
\end{theorem}

\begin{proof}[Proof of \Cref{thm:progress}]
Let $\{c_1,\dots,c_n\}$ be the unique base of $C_1$ ordered such that $\{c_1,\dots,c_k\}$ is an affinely-independent subset of maximal size $k \leq n$. In particular, $\{c_1,\dots,c_k\}$ is an affine basis of $\aff(C_1)$. Recall that $f$ extends to a unique surjective affine map $\overline{f}: \aff(C_1) \rightarrow \aff(C_2)$. It suffices to show that $\{\overline{f}(c_1), \dots, \overline{f}(c_k)\}$ is affinely-dependent. By assumption, there are distinct points $v_1, v_2 \in C_1$ with $f(v_1) = f(v_2)$. Since $C_1 \subseteq \aff(C_1)$, there are distinct, uniquely-determined vectors $\alpha,\beta \in \R^k$ such that $\sum_{i=1}^k \alpha_i = 1 = \sum_{i=1}^k \beta_i$, $v_1 = \sum_{i=1}^k \alpha_i c_i$, and $v_2 = \sum_{i=1}^k \beta_i c_i$. Then we have
\[
\sum_{i=1}^k \alpha_i \overline{f}(c_i) = \overline{f}\left(\sum_{i=1}^k \alpha_i c_i\right) = \overline{f}(v_1) = \overline{f}(v_2) = \overline{f}\left(\sum_{i=1}^k \beta_i c_i\right) = \sum_{i=1}^k \beta_i \overline{f}(c_i).
\]
Thus, $\{\overline{f}(c_1), \dots, \overline{f}(c_k)\}$ is affinely-dependent as desired.

\end{proof}

When a column is added that fixes a consistency defect, the convex algebra generated by the old table is a quotient of the convex algebra generated by the new table, where the two are not isomorphic. This means \Cref{thm:progress} applies and guarantees the new table generates a space with strictly higher dimension. Correctness of the algorithm follows trivially from termination and the characterization of the images of $\D(X)$ and $D(\bar X)$ under $\sem -$ in~\Cref{sec:lfp}.

\begin{example}[Reduced automaton Example]\label{ex:extra}
We show how the automaton from \Cref{ex:redvsminimal} (which we recall on the left below) can be reduced to an automaton with 4 states. In the middle is the consistent table corresponding to the automaton and how the rows in the table can be depicted to recover the convex combinations needed to obtain the reduced automaton on the right.

\medskip 
\noindent  \begin{tikzpicture}[baseline=(current bounding box.center),node distance = 1.75cm, arrows=->,font=\scriptsize]
    \tikzset{every state/.style={minimum size=0pt}}
        \node[state] (q1) {$q_1$};
        \node[state, right of=q1, accepting below, accepting text = {$1/2$}] (q2) {$q_2$};
        \node[state, below of=q1, accepting below, accepting text = {$1/2$}] (q3) {$q_3$};
        \node[state, right of=q3] (q4) {$q_4$};
        \node[state, below of=q3, accepting below, accepting text = {$1$}] (q5) {$q_5$};
        \node[state, right of=q5, accepting below, accepting text = {$1/4$}] (q6) {$q_6$};
        \draw (q1) edge[above] node[above]{$a,1$} (q2)
              (q2) edge[loop right] node[above]{$a,1$} (q2)
              (q3) edge[below] node[above]{$a,1$} (q4)
              (q4) edge[loop right] node[above]{$a,1$} (q4)
              (q5) edge[below] node[above]{$a,1$}(q6)
              (q6) edge[loop right] node[above]{$a,1$} (q6);
    \end{tikzpicture}
    \hspace{.7cm}
    \begin{tabular}{ c | c c }
        & $\eword$ & $a$ \\
        \hline
        $q_1$ & $0$ & $1/2$ \\
        $q_2$ & $1/2$ & $1/2$ \\
        $q_3$ & $1/2$ & $0$ \\
        $q_4$ & $0$ & $0$ \\
        $q_5$ & $1$ & $1/2$ \\
        $q_6$ & $1/4$ & $1/4$
    \end{tabular}
    \hspace{.7cm}
    \begin{tikzpicture}[baseline=(current bounding box.center), scale=1.75,font=\scriptsize]
        \coordinate (q1) at (0,0.5);
        \coordinate (q2) at (0.5,0.5);
        \coordinate (q3) at (0.5,0);
        \coordinate (q4) at (0,0);
        \coordinate (q5) at (1,0.5);
        \coordinate (q6) at (0.25,0.25);
        
        \draw[fill=black] (q1) circle (0.15em)
        node[above left] {$q_1$};
        \draw[fill=black] (q2) circle (0.15em)
        node[above right] {$q_2$};
        \draw[fill=black] (q3) circle (0.15em)
        node[below right] {$q_3$};
        \draw[fill=black] (q4) circle (0.15em)
        node[below left] {$q_4$};
        \draw[fill=black] (q5) circle (0.15em)
        node[below left] {$q_5$};
        \draw[fill=black] (q6) circle (0.15em)
        node[below left] {$q_6$};
        
        \draw[->] (0,0) -- (1.5, 0) node[right] {$\eword$};
        \draw[->] (0,0) -- (0, 1.5) node[above] {$a$};
        \draw[-] (0,1) -- (1,1);
        \draw[-] (1,1) -- (1,0);
    \end{tikzpicture}
    \begin{tikzpicture}[baseline=(current bounding box.center), node distance = 1.75cm, arrows=->,font=\scriptsize]
    \tikzset{every state/.style={minimum size=0pt}}
        \node[state] (q1) {$q_1$};
        \node[state, right of=q1, accepting below, accepting text = {$1/2$}] (q3) {$q_3$};
        \node[state, below of=q3] (q4) {$q_4$};
        \node[state, below of=q1, accepting below, accepting text = {$1$}] (q5) {$q_5$};
        \draw (q1) edge[above] node{$a,1/2$} (q3)
              (q1) edge[loop above] node{$a,1/2$} (q1)
              (q3) edge[right, bend left] node{$a,1$} (q4)
              (q4) edge[loop left] node{} (q4)
               (q5) edge[left] node{$a,1/2$} (q1)
              (q5) edge[left] node{$a,1/2$} (q3);
    \end{tikzpicture}
\end{example}

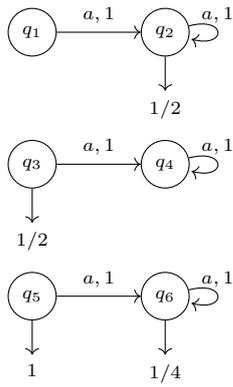
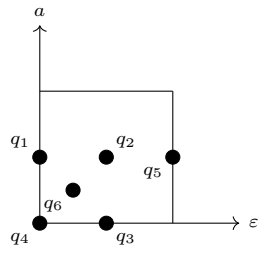
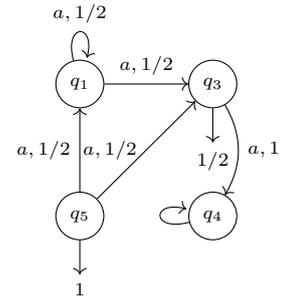
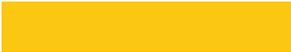
\end{document}